\documentclass[twocolumn,conference,letterpaper]{IEEEtran}
\usepackage[left=0.75in,right=0.75in,top=1in,bottom=0.75in]{geometry}
\usepackage{comment}

\usepackage{cite}
\usepackage{amssymb,amsmath,latexsym,amsfonts,amsthm,mathtools}

\usepackage{graphicx}
\usepackage{float,subcaption}
\usepackage{textcomp}
\usepackage{xcolor}
\definecolor{darkpastelpurple}{rgb}{0.59, 0.44, 0.84}
\usepackage{hyperref}
\hypersetup{colorlinks, breaklinks, citecolor=blue, linkcolor=blue, urlcolor=blue}
\usepackage[nameinlink]{cleveref}
\Crefformat{figure}{#2Fig.~#1#3}
\Crefmultiformat{figure}{Figs.~#2#1#3}{ and~#2#1#3}{, #2#1#3}{ and~#2#1#3}

\crefname{assumption}{assumption}{assumptions}  % singular, plural
\Crefname{assumption}{Assumption}{Assumptions} 
\usepackage{tikz}
\usetikzlibrary{automata, shapes, arrows, calc, arrows.meta, fit, positioning}

\theoremstyle{plain}
\newtheorem{theorem}{Theorem}

\newtheorem*{problem*}{Problem}
\theoremstyle{remark}
\newtheorem{remark}{Remark}
\newtheorem{assumption}{Assumption}

\newtheorem{definition}{Definition}
\theoremstyle{definition}

\usepackage{mathrsfs}

\usepackage{newtxmath}
\usepackage{booktabs}
\usepackage{duckuments}

\IEEEoverridecommandlockouts

\begin{document}
	% \title{Impact Time Control Using True Proportional Navigation Guidance with Bounded Input}
    \title{Bounded-Input True Proportional Navigation for Impact-Time Control}
    %\title{Nested Sliding Mode Based Impact Time Regulation Using True Proportional Navigation Guidance With Input Constraints}
	\author{Lohitvel Gopikannan, Shashi Ranjan Kumar,~\IEEEmembership{Senior Member, IEEE}, and Abhinav Sinha,~\IEEEmembership{Senior Member,~IEEE}
		\thanks{L. Gopikannan and S. R. Kumar are with the Intelligent Systems \& Control (ISaC) Lab, Department of Aerospace Engineering, Indian Institute of Technology Bombay, Mumbai 400076, India. (e-mails: 24m0023@iitb.ac.in, srk@aero.iitb.ac.in). A. Sinha is with the Guidance, Autonomy, Learning, and Control for Intelligent Systems (GALACxIS) Lab, Department of Aerospace Engineering and Engineering Mechanics, University of Cincinnati, OH 45221, USA. (e-mail: abhinav.sinha@uc.edu).}
	}

	\maketitle
	\thispagestyle{empty}
	
	\begin{abstract}
This paper proposes a nonlinear guidance strategy capable of intercepting a constant-velocity, non-maneuvering target while strictly satisfying the prescribed bounds on the control input (commanded acceleration). Unlike conventional strategies that estimate time-to-go using linearization or small-angle approximations, the proposed strategy employs true proportional-navigation guidance (TPNG) as a baseline, which utilizes an exact time-to-go formulation and is applicable over a wide range of target motions. In contrast to most existing strategies, which do not incorporate control input bounds into the guidance design, the proposed approach explicitly accounts for these limits by modeling the interceptor acceleration as a dynamic variable. Based on the sliding mode control technique, an effective guidance law that achieves time-constrained interception while accounting for bounded input is then derived. The performance of the proposed strategy is evaluated for various engagement scenarios.
\end{abstract}
	
	\begin{IEEEkeywords}
    Impact time Guidance, Bounded Input, Moving Target Interception, Uncrewed Autonomous Vehicles.
	\end{IEEEkeywords}
	
	\section{Introduction}\label{sec:introduction}
The increasing complexity of modern engagement scenarios necessitates the development of advanced guidance strategies capable of ensuring precise interception while simultaneously satisfying stringent terminal constraints, such as prescribed impact angle \cite{lin2020impact,kumar2019finite} and impact time \cite{jeon2006impact,jeon2016impact}. Precise control of impact time is particularly important in coordinated salvo attacks \cite{zhou2016}, where the simultaneous arrival of multiple interceptors is essential to saturate defensive systems such as close-in weapon systems, which can engage only a limited number of threats at any given time.

Although one of the applications of impact-time control lies in cooperative multi-interceptor scenarios, the fundamental problem can be effectively formulated and analyzed within a single interceptor-target engagement. Early efforts toward impact time regulation focused on stationary target interception, where \cite{jeon2006impact} proposed a guidance strategy combining a proportional-navigation (PN) guidance term with an error feedback term to achieve interception at a prescribed time within a linearized engagement framework. A nonlinear extension of this approach was subsequently presented in \cite{jeon2016impact}.  The work in \cite{cho2016modified} obtained an exact closed-form time-to-go for  pure PN guidance applicable to all guidance gains greater than 1 and all initial heading errors under nonlinear engagement kinematics, and proposed a time-varying gain impact time control law. Look-angle shaping-based guidance laws were proposed in \cite{tekin2017polynomial,tekin2017adaptive,tekin2018impact}, employing arbitrary-order polynomial functions to control the impact time. Nevertheless, higher-order polynomial functions introduce increased computational demands that may complicate practical implementation. Lyapunov-based impact-time guidance strategies were studied in \cite{yanushevsky2005new,kim2015lyapunov,saleem2016lyapunov}. However, the range of attainable impact times is limited as the look angle decreases monotonically during the engagement. Sliding mode control (SMC) has also been extensively employed in the design of impact-time guidance laws \cite{kumar2016impact,cho2016nonsingular,hu2019sliding} owing to its robustness against modeling uncertainties and external disturbances. The authors in \cite{kumar2016impact} constructed a sliding surface based on the impact time error and line-of-sight (LOS) rate, which was subsequently simplified in \cite{cho2016nonsingular} to incorporate only the impact time error. A super-twisting sliding mode–based guidance law was proposed in \cite{sinha2020super} without small-angle assumptions, ensuring finite-time convergence, robustness to uncertainties, and effectiveness even under large initial heading errors. The majority of the above-mentioned guidance strategies were derived for stationary targets, whereas \cite{jeon2016impact,kumar2016impact,sinha2020super,doi:10.2514/1.G007122} employed the predicted interception point (PIP) approximation for moving targets. This approach, however, requires approximating the target's future position for moving targets, potentially introducing estimation errors that compromise interception performance.

For interception of moving targets, a terminal SMC-based strategy was presented in \cite{hu2019sliding}; however, the method requires offline parameter tuning, making it computationally intensive. The authors in \cite{tekin2016control} proposed a structured multi-phase guidance scheme without any time-to-go information. As an alternative to PN guidance, a deviated pursuit-based impact time guidance law was proposed in \cite{kumar2019deviated}, offering an exact time-to-go formulation. This was subsequently extended to incorporate interceptor dynamics in \cite{sinha2021impact,sinha2021nonsingular}. However, the deviated pursuit strategy is only applicable when the interceptor has a speed advantage over the target and is ineffective against stationary targets. It should be noted that most of the aforementioned strategies are applicable only to constant-velocity vehicles or those in which control is achieved solely through lateral acceleration. The work in \cite{kumar2022true} addressed this limitation by proposing a TPNG-based impact time guidance law with acceleration components in both radial and tangential directions, extending applicability to a broader class of target motions. The time-critical TPNG was then discussed in \cite{10.1115/1.4066259} under optimal time-to-go error dynamics.

From the above discussions, it is clear that significant work has been done toward intercepting stationary and moving targets at a prespecified impact time. However, one common limitation of these strategies is that they are designed under the assumption of unbounded control authority, which is unrealistic given the physical limitations of the interceptor's actuators. When the commanded control exceeds practically achievable limits, actuator saturation occurs, degrading overall closed-loop performance as discussed in \cite{dong2018guidance} and potentially leading to mission failure. Although actuator constraints are commonly handled through static nonlinearities in most of the works, this approach results in sluggish response and degraded performance, without guaranteeing closed-loop stability.

Motivated by the above limitations, this paper proposes a nonlinear guidance strategy that explicitly incorporates bounds on available control input within the guidance design for the interception of a constant-velocity, non-maneuvering target at a prespecified impact time. Unlike most of the existing approaches, the proposed method is applicable to vehicles capable of generating control inputs in both radial and tangential directions, thereby accommodating a broader class of autonomous platforms. The resulting guidance law ensures stable closed-loop behavior and accurate interception at the desired impact time, enhancing reliability and practical applicability under realistic actuator constraints.

\section{Background and Problem Formulation}
\begin{figure}[h!]
    \centering
    \includegraphics[width=0.9\linewidth]{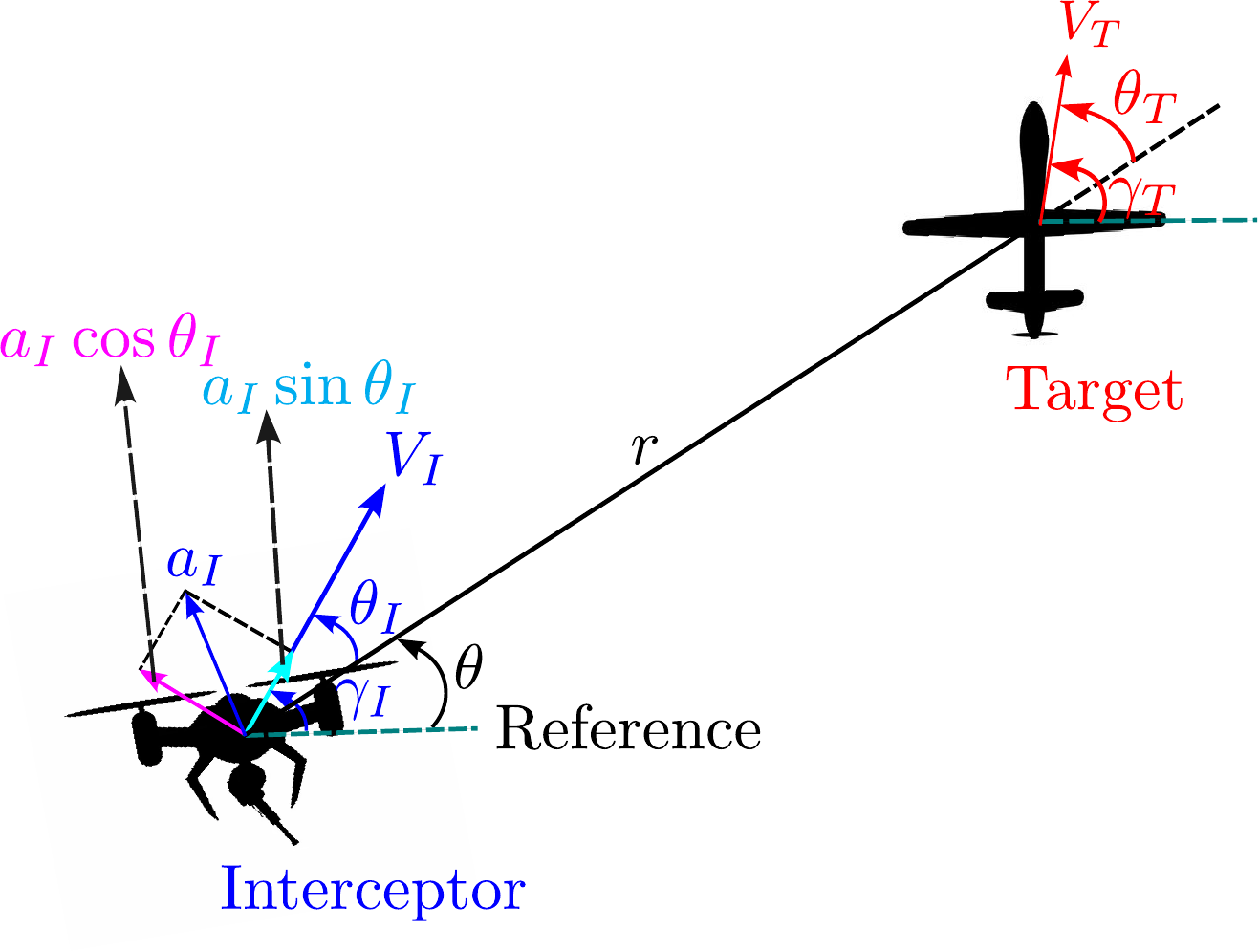}
    \caption{Planar interceptor-target engagement}
    \label{fig:enggeo}
\end{figure}
We consider a planar interception problem in which an uncrewed autonomous vehicle (UAV) engages a constant-velocity, non-maneuvering target as depicted in \Cref{fig:enggeo}. The interceptor and target speeds are defined as $V_I$ and $V_T$, respectively, whereas $\gamma_I$ and $\gamma_T$ denote their corresponding flight path angles. The relative separation between the interceptor and target is defined as $r$, with its LOS oriented at an angle $\theta$ from the reference. The interceptor adjusts its trajectory by generating an acceleration command $a_I$, which is applied perpendicular to LOS. The nonlinear relative kinematic equations governing the engagement dynamics are:
\begin{subequations}
\begin{align}
% \dot{r} &= V_{T}\cos(\gamma_{T} - \theta) 
%               - V_{M}\cos(\gamma_{M} - \theta) 
%               = V_{r},\\
% r\,\dot{\theta} &= V_{T}\sin(\gamma_{T} - \theta) 
%               - V_{M}\sin(\gamma_{M} - \theta) 
%               = V_{\theta}, \\
% \dot{\gamma}_{M} &= \frac{a_{M}\cos(\gamma_{M} - \theta)}{V_{M}}, \\
% \dot{V}_{M} &= a_{M}\sin(\gamma_{M} - \theta).
\dot{r} &= V_{T}\cos\theta_T 
              - V_{I}\cos\theta_I 
              = V_{r},\\
r\,\dot{\theta} &= V_{T}\sin\theta_T  
              - V_{I}\sin\theta_I
              = V_{\theta}, \\
\dot{\gamma}_{I} &= \frac{a_{I}\cos\theta_I}{V_{I}}, \\
\dot{V}_{I} &= a_{I}\sin\theta_I.\label{eq:vmdot}
\end{align}
\label{eq:enggeo}
\end{subequations}
where $\theta_I = \gamma_I - \theta$ and $\theta_T = \gamma_T - \theta$ represent the lead angles of the interceptor and target, respectively, and $V_r$ and $V_\theta$ correspond to the radial and tangential components of the relative velocity with respect to the LOS, satisfying
\begin{subequations}\label{eq:RV_dot}
\begin{align}
    \dot V_r&=\dot \theta V_\theta\\
    \dot V_\theta &=-\dot \theta V_r-a_I
\end{align}
\end{subequations}

As evident from \eqref{eq:vmdot}, the interceptor's speed changes with time due to the longitudinal acceleration component. This characteristic makes the formulation well-suited for a large class of autonomous vehicles such as quadcopters and uncrewed surface vessels (USVs) \cite{verma2025true}, which can modulate their thrust or propulsion to adjust both speed and trajectory during maneuvers. Before proceeding further, the key assumptions and definitions adopted in this study are stated as follows.
\begin{assumption}
The target and interceptor are treated as ideal point-mass vehicles. Interceptor's dynamics is neglected during the guidance design, assuming a sufficiently fast autopilot response.
\end{assumption}
\begin{assumption}
Range and LOS information are available to the interceptor throughout engagement. 
\end{assumption}
\begin{definition}
For an interceptor-target engagement, the impact time $t_{d}$ denotes the total duration from the interceptor's launch to the instant of target interception. This can be expressed  mathematically as $t_{d} = t_{\mathrm{el}} + t_{\mathrm{go}}$, where $t_{\mathrm{el}}$ denotes the elapsed time since launch and $t_{\mathrm{go}}$ represents the remaining time until interception.\end{definition}
We are now in a position to formally define the problem addressed in this work.
\begin{problem*}
For the planar interceptor-target engagement depicted in \Cref{fig:enggeo}, the objective of this paper is to develop a  nonlinear guidance strategy that guarantees interception of a constant velocity non-maneuvering target while satisfying the following constraints: (i) $r\to0 \text{ as } t\to t_{d}$ (target interception), and (ii) $a_{\min}<a_I<a_{\max}~ \forall\; t\in [0,t_{d}]$ (bounds on control input), where $a_{\max}>0$ and $a_{\min}<0$.
% \begin{itemize}
%     \item Target interception should occur at a desired time $t_d$ specified prior to the engagement, i.e.,
%     \begin{align}
%         r\to0\quad \text{as}\quad t\to t_d
%     \end{align}
%     \item The interceptor's acceleration $a_I$ must remain within prescribed bounds throughout the engagement. Specifically, it should satisfy 
%        \begin{align}
%            a_{\min}<a_I<a_{\max}\quad \forall\; t\in [0,t_d]
%        \end{align}
%        where $a_{\max}>0$ and $a_{\min}<0$.
% \end{itemize}
\end{problem*}
\begin{remark}
Most of the existing guidance strategies employ symmetric bounds on lateral acceleration. However, in practice, actuator systems exhibit asymmetric control capabilities due to direction-dependent limitations. In addition, safety and operational considerations may require conservative bounds in certain directions while permitting larger control authority in others. As a result, asymmetric constraints more accurately reflect realistic actuation limits.
\end{remark}
% For the planar engagement geometry depicted in \Cref{fig:enggeo}., the objective is to develop a nonlinear guidance strategy that guarantees interception at the desired impact time $t_d$ while respecting bounds on the interceptor's acceleration.
In this paper, TPNG is adopted as a baseline guidance strategy due to the availability of acceleration components in both the axial and radial directions relative to the velocity vector, and its applicability to a wide range of target motions. 
\section{Design of Guidance Strategy}
This section first introduces the input saturation model, followed by the design of a guidance law that satisfies the prescribed constraints. Existing approaches commonly handle acceleration constraints by applying saturation functions during implementation. However, overall system stability is not guaranteed as the underlying analyses assume unbounded control inputs. In contrast, the proposed guidance strategy explicitly incorporates bounded acceleration while ensuring closed-loop stability. To achieve this, and motivated by the approach in \cite{kumar2025provably}, $a_I$ is modeled using a first-order asymmetric smooth input saturation model as
\begin{align}\label{eq:saturation_model}
    \dot a_{I}&=\left[\mu\left\{1-\left(\dfrac{a_I}{a_{\max}}\right)^\rho\right\}+(1-\mu)\left\{1-\left(\dfrac{a_I}{a_{\min}}\right)^\rho\right\}\right]a_I^c-\lambda a_I
\end{align}
where $a_I^c$ is the commanded input from the guidance strategy and $a_I$ is the achieved acceleration, $\mu=\begin{cases}
    1\quad a_I\geq 0\\0\quad a_I<0
\end{cases}$, $\lambda\in(0,1)$and $\rho=2n$ where $n\in \mathbb{N}$. The detailed analysis demonstrating the confinement of $a_I$ within the prescribed bounds $\forall t$ is presented in the following theorem.
\begin{theorem}\label{th:theorem1}
Under the saturation model \eqref{eq:saturation_model}, finiteness of the commanded input $a_I^c$ guarantees that $a_I$ is constrained to the set $\mathbb{S}\coloneqq\{a_I\in \mathbb{R}~|\;a_{\min}<a_I<a_{\max}\}$  throughout the engagement.
\end{theorem}
\begin{proof}
    The finiteness of the commanded input implies that $|a_{I}^c|\leq\chi\subseteq \mathbb{R}$,$\forall \;t\geq0$. If $a_I=a_{\max}$ then from \eqref{eq:saturation_model} we have $\dot a_I=-\lambda a_{\max}$, consequently if $a_I$ reaches $a_{\max}$ then it will start decreasing.  Similarly, if $a_I=a_{\min}$ from \eqref{eq:saturation_model} we have $\dot a_I=-\lambda a_{\min}$, since $a_{\min}<0$, $a_{I}$ will start on increasing once it reaches $a_{\min}$. Now will establish that $a_I$ never reaches $a_{\max}$ or $a_{\min}$. Consider the case of $a_I\in[0, a_{\max})$, then one can write \eqref{eq:saturation_model} as
    \begin{align}\label{eq:satuaration_model_right}
        \dot a_I&=\left[1-\left(\dfrac{a_I}{a_{\max}}\right)^\rho\right]a_I^c-\lambda a_I.
    \end{align}
It is clearly evident from \eqref{eq:satuaration_model_right} that $a_I$ can increase only iff $a_I^c>0$, since $\left[1-\left(\dfrac{a_I}{a_{\max}}\right)^\rho\right]\geq0$ for  $a_I\in[0, a_{\max})$. Let us assume that $a_I^c>0$, then using the fact that $|a_I^c|\leq\chi$, one can modify \eqref{eq:satuaration_model_right} as
\begin{align}
    \dot a_I&\leq\left[1-\left(\dfrac{a_I}{a_{\max}}\right)^\rho\right]\chi-\lambda a_I\nonumber\\
           &=\chi\left[1-\left(\dfrac{a_I}{a_{\max}}\right)^\rho-\dfrac{\lambda a_{\max}}{\chi}\left(\dfrac{a_I}{a_{\max}}\right)\right].
\end{align}
Since $-\left(\dfrac{a_I}{a_{\max}}\right)^\rho\geq-\left(\dfrac{a_I}{a_{\max}}\right)$, the above equation can be simplified to
\begin{align}
\dot a_I&\leq\chi\left[1-\left(\dfrac{a_I}{a_{\max}}\right)^\rho\left(1+\dfrac{\lambda a_{\max}}{\chi}\right)\right].
\end{align}
From the above equation, it is clear that $\dot a_I\leq0$ iff $\left[1-\left(\dfrac{a_I}{a_{\max}}\right)^\rho\left(1+\dfrac{\lambda a_{\max}}{\chi}\right)\right]\leq0$. On solving for $a_I$, one can obtain $a_I\geq a_{\max}\left[\dfrac{\chi}{\chi+\lambda a_{\max}}\right]^{\frac{1}{\rho}}:=\tilde a_{\max}<a_{\max}$. This implies that for $a_I\geq\tilde a_{\max}$, $\dot a_I\leq0$, therefore we can conclude  that $a_I\leq {\tilde a}_{\max}<a_{\max}$. Now consider the case of $a_I<0$, then \eqref{eq:saturation_model} reduces to
    \begin{align}\label{eq:satuaration_model_left}
        \dot a_I&=\left[1-\left(\dfrac{a_I}{a_{\min}}\right)^\rho\right]a_I^c-\lambda a_I.
    \end{align}
It has been established earlier that $a_I$ will start increasing once it reaches $a_{\min}$, therefore consider the case of $a_I\in(a_{\min},0)$. Since $\left[1-\left(\dfrac{a_I}{a_{\min}}\right)^\rho\right]\geq0$, one can observe from \eqref{eq:satuaration_model_left} that $a_I$ will decrease iff $a_I^c<0$, Without loss of generality assume $a_I^c$, then using the fact that $-\chi\leq a_I^c$ one can write \eqref{eq:satuaration_model_left} as
\begin{align}
    \dot a_I&\geq-\left[1-\left(\dfrac{a_I}{a_{\min}}\right)^\rho\right]\chi-\lambda a_I\nonumber\\
           &=\chi\left[-1+\left(\dfrac{a_I}{a_{\min}}\right)^\rho-\dfrac{\lambda a_{\min}}{\chi}\left(\dfrac{a_I}{a_{\min}}\right)\right].
\end{align}
Since $a_I<0$, $\dfrac{\lambda  {a}_{\min}}{\chi} \left(\dfrac{a_I}{a_{\min}}\right)^\rho\leq-\dfrac{\lambda  {a}_{\min}}{\chi}\left(\dfrac{a_I}{a_{\min}}\right)$, the above equation can be simplified as
\begin{align}
\dot a_I&\geq\chi\left[-1+\left(\dfrac{a_I}{a_{\min}}\right)^\rho\left(1+\dfrac{\lambda a_{\min}}{\chi}\right)\right].
\end{align}
It is clear from above equation that $\dot a_I$ can increase iff $\left[-1+\left(\dfrac{a_I}{a_{\min}}\right)^\rho\left(1+\dfrac{\lambda a_{\min}}{\chi}\right)\right]\geq0$. On solving for $a_I$, one can obtain $a_I\geq a_{\min}\left[\dfrac{\chi}{\chi+\lambda a_{\min}}\right]^{\frac{1}{\rho}}\coloneqq \tilde a_{\min}>a_{\min} $. This implies that for $a_I\leq a_{\min}$, $\dot a_I\geq0$, which further implies that $a_I\geq\tilde a_{\min}>a_{\min}$. Therefore, from the above analysis, we can conclude that acceleration $a_I$ is confined to the set $\mathbb{S}\coloneqq\{a_I\in \mathbb{R}~|\;a_{\min}<a_I<a_{\max}\}$  throughout the engagement. This concludes the proof.
\end{proof}
In the analysis that follows, the aim is to design the commanded acceleration $a_I^c$ that guarantees the desired objectives. For an interceptor guided by TPNG, the time-to-go expression is given by
\begin{equation}\label{eq:tgo}
t_{\mathrm{go}}= - \frac{r \left(V_{r} + 2\,c \right)}{V_{\theta}^{2} + V_{r}^{2} + 2\,c\,V_{r}}.
\end{equation}
As established  in \cite{kumar2022true}, successful target interception is guaranteed iff the proportionality factor satisfies $ c \gg \dfrac{V_I+V_T}{2}$. Differentiating \eqref{eq:tgo}, one obtains
\begin{align}\label{eq:tgodynamics}
    \dot t_{\mathrm{go}}= -1+\dfrac{2c(V_r+2c)V_\theta^2}{\left(V_\theta^2+V_r^2+2cV_r\right)^2}-\dfrac{2(V_r+2c)V_\theta r}{\left(V_\theta^2+V_r^2+2cV_r\right)^2}a_I.
\end{align}
Let us define the impact time error as the difference between the time-to-go estimate and its desired value. Mathematically, it can be expressed as $e=t_{\mathrm{go}}-t_\mathrm{go}^d$ where $t_\mathrm{go}^d=t_d-t_\mathrm{{el}}$. On differentiating this error in time-to-go with respect to time and utilizing \eqref{eq:tgodynamics}, one can obtain
\begin{align}\label{eq:error_dynamics}
    \dot e&=1+\dot t_{\mathrm{go}} \nonumber\\
          &=\dfrac{2c(V_r+2c)V_\theta^2}{\left(V_\theta^2+V_r^2+2cV_r\right)^2}-\dfrac{2(V_r+2c)V_\theta r}{\left(V_\theta^2+V_r^2+2cV_r\right)^2}a_I\nonumber\\
          &=F+Ba_I.
\end{align}
where $F=\dfrac{2c(V_r+2c)V_\theta^2}{\left(V_\theta^2+V_r^2+2cV_r\right)^2}$ and $B=-\dfrac{2(V_r+2c)V_\theta r}{\left(V_\theta^2+V_r^2+2cV_r\right)^2}$. On further differentiating \eqref{eq:error_dynamics} with respect to time, one can obtain 
\begin{align}\label{eq:eddot_1}
    \ddot e=\dot F+ \dot Ba_I+B\dot a_I.
\end{align}
For brevity, the differentiation is performed term-wise, and the resulting components are combined at the end. On differentiating $F$ with respect to time, one can get
\begin{align}
    \dot F&= \dfrac{2c}{\left(V_\theta^2+V_r^2+2cV_r\right)^4}\left[\left(V_\theta^2 \dot V_r+2V_\theta \dot V_\theta\left(V_r+2c\right)\right)\right.\nonumber \\&\left(V_\theta^2+V_r^2+2cV_r\right)^2
    \left.-4V_\theta^2\left(V_r+2c\right)\left(V_\theta^2+V_r^2+2cV_r\right)\right.\nonumber\\
    &\left.\times\left(V_\theta \dot V_\theta+V_r \dot V_r+c\dot V_r\right)\right].
\end{align}
Simplifying the above expression yields,
\begin{align}
    \dot F&=\dfrac{2c\left(V_\theta^2 \dot V_r+2V_\theta \dot V_\theta(V_r+2c)\right)}{\left(V_\theta^2+V_r^2+2cV_r\right)^2}\nonumber\\
    &-\dfrac{8c(V_r+2c)V_\theta^2\left(V_\theta \dot V_\theta+V_r \dot V_r+c\dot V_r\right)}{\left(V_\theta^2+V_r^2+2cV_r\right)^3}\label{eq:F_dot}
\end{align}
Similarly, on differentiating $B$ with respect to time, one can obtain
\begin{align}
    \dot B&=\dfrac{-2}{\left(V_\theta^2+V_r^2+2cV_r\right)^4} \left[\left(r V_\theta \dot V_r+ \left(V_rV_\theta+r\dot V_\theta\right)\left(V_r+2c\right)\right)\right. \nonumber \\
    &\left.\left(V_\theta^2+V_r^2+2cV_r\right)^2-4V_\theta r\left(V_r+2c\right)\left(V_\theta^2+V_r^2+2cV_r\right)\right.\nonumber\\
    &\left.\times\left(V_\theta \dot V_\theta+V_r \dot V_r+c\dot V_r\right)\right].\label{eq;B_dot}
\end{align}
Simplifying the above equation results in 
\begin{align}
    \dot B&=-\dfrac{2\left(r V_\theta \dot V_r+ (V_rV_\theta+r\dot V_\theta)(V_r+2c)\right)}{\left(V_\theta^2+V_r^2+2cV_r\right)^2} \nonumber\\
        &+\dfrac{8(V_r+2c)V_\theta r\left(V_\theta \dot V_\theta+V_r \dot V_r+c\dot V_r\right)}{\left(V_\theta^2+V_r^2+2cV_r\right)^3}.
\end{align}
On rewriting \eqref{eq:eddot_1} using \eqref{eq:RV_dot},\eqref{eq:saturation_model},\eqref{eq:F_dot} and \eqref{eq;B_dot} yields
% \begin{align}
%        \ddot e&= \dfrac{2c\dot \theta V_\theta\left[V_\theta^2-2V_r^2-4cV_r\right]}{\left(V_\theta^2+V_r^2+2cV_r\right)^2}+\dfrac{-4cV_\theta(V_r+2c)a_I}{\left(V_\theta^2+V_r^2+2cV_r\right)^2}\nonumber\\
%        &-\dfrac{8c^2\left(V_r+2c\right)V_\theta^3\dot \theta }{\left(V_\theta^2+V_r^2+2cV_r\right)^3}+\dfrac{8c\left(V_r+2c\right)V_\theta^3a_I}{\left(V_\theta^2+V_r^2+2cV_r\right)^3}\nonumber\\
%     &-\dfrac{2a_I\left[V_\theta^3-r\left(V_r+2c\right)a_I\right]}{\left(V_\theta^2+V_r^2+2cV_r\right)^2}+\dfrac{8V_\theta^2a_I(V_r+2c)\left(cV_\theta-ra_I\right)}{\left(V_\theta^2+V_r^2+2cV_r\right)^3}\nonumber\\
%     &+\dfrac{2(V_r+2c)V_\theta r\lambda a_I}{\left(V_\theta^2+V_r^2+2cV_r\right)^2}-\dfrac{2(V_r+2c)V_\theta r}{\left(V_\theta^2+V_r^2+2cV_r\right)^2}\left[\mu\left\{1-\left(\dfrac{a_I}{a_{\max}}\right)^\rho\right\}\right. \nonumber\\
%     &\left.+(1-\mu)\left\{1-\left(\dfrac{a_I}{a_{\min}}\right)^\rho\right\}\right]a_I^c\nonumber
% \end{align}
\begin{align}
       \ddot e&= \dfrac{2c\dot \theta V_\theta\left[V_\theta^2-2V_r^2-4cV_r\right]}{\left(V_\theta^2+V_r^2+2cV_r\right)^2}+\dfrac{-4cV_\theta(V_r+2c)a_I}{\left(V_\theta^2+V_r^2+2cV_r\right)^2}\nonumber\\
       &-\dfrac{8c\left(V_r+2c\right)V_\theta^3\left[ c\dot \theta-a_I\right] }{\left(V_\theta^2+V_r^2+2cV_r\right)^3}-\dfrac{2a_I\left[V_\theta^3-r\left(V_r+2c\right)a_I\right]}{\left(V_\theta^2+V_r^2+2cV_r\right)^2}\nonumber\\
    &+\dfrac{8V_\theta^2a_I(V_r+2c)\left(cV_\theta-ra_I\right)}{\left(V_\theta^2+V_r^2+2cV_r\right)^3}+\dfrac{2(V_r+2c)V_\theta r\lambda a_I}{\left(V_\theta^2+V_r^2+2cV_r\right)^2}\nonumber\\
    &-\dfrac{2(V_r+2c)V_\theta r}{\left(V_\theta^2+V_r^2+2cV_r\right)^2}\left[\mu\left\{1-\left(\dfrac{a_I}{a_{\max}}\right)^\rho\right\}+(1-\mu)\right.\nonumber \\&\times\left.\left\{1-\left(\dfrac{a_I}{a_{\min}}\right)^\rho\right\}\right]a_I^c\nonumber\\
    &=  \mathcal{F^\star} +  \mathcal{B^\star} a_I^c ,\label{eq:eddot_2}
\end{align}
where $\mathcal{F^\star}$ and $\mathcal{B^\star}$ are defined as
\begin{align}
    \mathcal{F^\star}&=\dfrac{2c\dot \theta V_\theta\left[V_\theta^2-2V_r^2-4cV_r\right]}{\left(V_\theta^2+V_r^2+2cV_r\right)^2}+\dfrac{-4cV_\theta(V_r+2c)a_I}{\left(V_\theta^2+V_r^2+2cV_r\right)^2}\nonumber\\
       &-\dfrac{8c\left(V_r+2c\right)V_\theta^3\left[ c\dot \theta-a_I\right] }{\left(V_\theta^2+V_r^2+2cV_r\right)^3}-\dfrac{2a_I\left[V_\theta^3-r\left(V_r+2c\right)a_I\right]}{\left(V_\theta^2+V_r^2+2cV_r\right)^2}\nonumber\\
    &+\dfrac{8V_\theta^2a_I(V_r+2c)\left(cV_\theta-ra_I\right)}{\left(V_\theta^2+V_r^2+2cV_r\right)^3}+\dfrac{2(V_r+2c)V_\theta r\lambda a_I}{\left(V_\theta^2+V_r^2+2cV_r\right)^2}\\
    \mathcal{B^\star}&=-\dfrac{2(V_r+2c)V_\theta r}{\left(V_\theta^2+V_r^2+2cV_r\right)^2}\left[\mu\left\{1-\left(\dfrac{a_I}{a_{\max}}\right)^\rho\right\}+(1-\mu)\right.\nonumber \\&\times\left.\left\{1-\left(\dfrac{a_I}{a_{\min}}\right)^\rho\right\}\right]
\end{align}
It can be observed from \eqref{eq:eddot_2} that the error dynamics exhibits relative degree two with respect to $a_I^c$.  The guidance design problem, therefore, reduces to stabilizing the second-order nonlinear system \eqref{eq:eddot_2} through a suitable design of $a_I^c$. 
\begin{remark}
Although finite-time, fixed-time, and prescribed-time control techniques exist for stabilizing second-order systems, input constraints inherently introduce a minimum achievable settling time because the system cannot be driven faster than the available control authority permits. As a result, both lower and upper bounds on the settling time arise in practice. The exact characterization of this minimum time remains a challenging problem and is part of our further research.  Although asymptotic schemes are also subject to these limitations, they are relatively simpler to design and require less aggressive tuning. Therefore, an asymptotic formulation is adopted as a practical starting point.
\end{remark}
Consider the sliding surface,
\begin{align} \label{eq:sliding_manifold}
    S=\dot e+\alpha e.
\end{align}
On differentiating \eqref{eq:sliding_manifold} with respect to time one may obtain
\begin{align}\label{eq:S_dot}
    \dot S= \ddot e+\alpha \dot e= \mathcal{F^\star}+\alpha\left(F+Ba_I\right)+\mathcal{B^\star}a_I^c.
\end{align}
Now, we design a feedback control law $a_I^c$ that enables time-constrained interception while satisfying 
the bounded input constraint, whose formulation is presented next.
\begin{theorem}
    For the planar engagement problem under consideration, whose engagement kinematics is governed by \eqref{eq:enggeo}, if the guidance command (commanded input) is designed as 
    \begin{align}\label{eq:a_M}
        a_I^c=\dfrac{-\mathcal{F^\star}-\alpha\left(F+Ba_I\right)-\dfrac{\mathcal{M}}{g(S)}\operatorname{sign}S}{\mathcal{B^\star
        }}
    \end{align}
 where $g(S)=\theta+(1-\theta) e^{-\kappa|S|^\eta}$, $\eta$ is a strictly positive integer, $k\in \mathbb{R}_{+}$ and $0<\theta<1$, then sliding mode is enforced on the chosen surface $S$. Consequently, the time-to-go error asymptotically converges to zero, thereby achieving time-constrained interception under bounded control input.
\end{theorem}
\begin{proof}
    Consider the Lyapunov function candidate as $V=\dfrac{S^2}{2}$. On differentiating it with respect to time and utilizing \eqref{eq:S_dot}, one can get
    \begin{align}
        \dot V=S\dot S=S\left[ \mathcal{F^\star}+\alpha\left(F+Ba_I\right)+\mathcal{B^\star}a_I^c\right].
    \end{align}
On substituting \eqref{eq:a_M}, the above equation reduces to 
\begin{align}
    \dot V=-\dfrac{\mathcal{M}}{g(S)}|S|.
\end{align}
Since $g(S)$ is strictly positive function, $\dot V<0$ for all $S\neq 0$. Therefore, when the sliding mode $S$ is enforced, i.e., $S=0$, then \eqref{eq:sliding_manifold} reduces to $\dot e=-\alpha e$, which ensures that error in time-to-go $e$ and its rate $\dot e$ will converge to zero asymptotically, thereby leading to a time-constrained interception while satisfying input constraints. This concludes the proof.
\end{proof}
\begin{remark}
    Consider the case when $|S|$ increases, then it can be observed that $g(S)\to\theta$ which implies that $\frac{\mathcal{M}}{g(S)}\to\frac{\mathcal{M}}{\theta}>\mathcal{M}$. This indicates that, during the reaching phase, the effective control gain increases, resulting in a stronger attraction of the errors $e$ and $\dot e$ towards $S$, thereby promoting faster convergence. As $|S|$ decreases, $g(S)$ approaches unity and $\frac{\mathcal{M}}{g(S)} \to \mathcal{M}$. This implies that, in the vicinity of the sliding surface $S$, the effective control gain is gradually reduced.
\end{remark}
On substituting \eqref{eq:a_M} into \eqref{eq:S_dot}, it is evident that the reaching law used to drive the system to the sliding surface $S$ is
\begin{align}\label{eq:reaching_law}
    \dot{S} = -\dfrac{\mathcal{M}}{g(S)}\operatorname{sign}(S),
\end{align}
which differs from the standard reaching law  $\dot{S} = -\mathcal{M}\operatorname{sign}(S)$. Integration of this reaching law yields the reaching time $t_{r_1} = \dfrac{|S(0)|}{\mathcal{M}}$ to the sliding surface $S$. Under the considered reaching law, let $t_{r_2}$ be the reaching time to the sliding surface $S$. One can write \eqref{eq:reaching_law} as
\begin{align}
\dot{S}\left[\theta + (1-\theta)e^{-\kappa |S|^{\eta}}\right]
= -\mathcal{M}\,\operatorname{sign}(S).
\end{align}
Integrating the above equation from $0$ to $t_{r_2}$ and considering $S(t_{r_2})=0$, one can obtain
\begin{align}\label{eq:tr_2_prelim}
t_{r_2} &= \frac{1}{\mathcal{M}}
\left[
\theta |S(0)|
+ (1-\theta)
\int_{0}^{S(0)}
\operatorname{sign}(S)e^{-\kappa |S|^{\eta}}\, dS
\right].
\end{align}
Consider the case when $S\leq0$ for $t\leq t_{r_2}$, then one can write
\begin{align}\label{eq:left}
\int_{0}^{S(0)}
\operatorname{sign}(S)e^{-\kappa |S|^{\eta}}\, dS
&= -\int_{0}^{S(0)} e^{-\kappa |S|^{\eta}}\, dS \nonumber\\
&= \int_{0}^{-S(0)} e^{-\kappa |S|^{\eta}}\, dS.
\end{align}
Similarly, if $S\geq0$ for $t\leq t_{r_2}$, then 
\begin{align}\label{eq:right}
\int_{0}^{S(0)}
\operatorname{sign}(S)e^{-\kappa |S|^{\eta}}\, dS
= \int_{0}^{S(0)}
e^{-\kappa |S|^{\eta}}\, dS.
\end{align}
Combining \eqref{eq:left} and \eqref{eq:right}, one can write
\begin{align}
    \int_{0}^{S(0)}
\operatorname{sign}(S)e^{-\kappa |S|^{\eta}}\, dS=\int_{0}^{|S(0)|}
e^{-\kappa |S|^{\eta}}\, dS
\end{align}
 Substituting this result in \eqref{eq:tr_2_prelim} yields
 \begin{align}\label{eq:tr2}
t_{r_2} &= \frac{1}{\mathcal{M}}
\left[
\theta |S(0)|
+ (1-\theta)
\int_{0}^{|S(0)|}
e^{-\kappa |S|^{\eta}}\, dS
\right]. \end{align}
To facilitate a meaningful comparison between the standard
reaching law and the considered reaching law, subtracting $t_{r_1}$
from the above expression yields
\begin{align}\label{eq:tr_diff}
    t_{r_2}-t_{r_1}=\dfrac{(1-\theta)}{\mathcal{M}}\int_0^{|S(0)|} \left[e^{-\kappa|S|^\eta}-1\right]dS
\end{align}
Since $e^{-\kappa|S|^\eta}-1<0$ always, this implies that the considered reaching law reaches faster than the standard reaching law to the sliding surface $S$.
\begin{comment}
    Upon integrating \eqref{eq:reaching_law}, and utilizing $t_{r_1}$, one can obtain (for brevity, detailed steps are omitted)
    %begin{align}
    t_{r_2}-t_{r_1}=\dfrac{(1-\theta)}{\mathcal{M}}\int_0^{|S(0)|} \left[e^{-\kappa|S|^\eta}-1\right]dS
    %\end{align}
Since $e^{-\kappa|S|^\eta}-1<0$ always, this implies that considered reaching law reaches faster than standard reaching law to sliding surface $S$.
\end{comment}

\begin{figure}[h!]
    \centering
    \includegraphics[width=\linewidth]{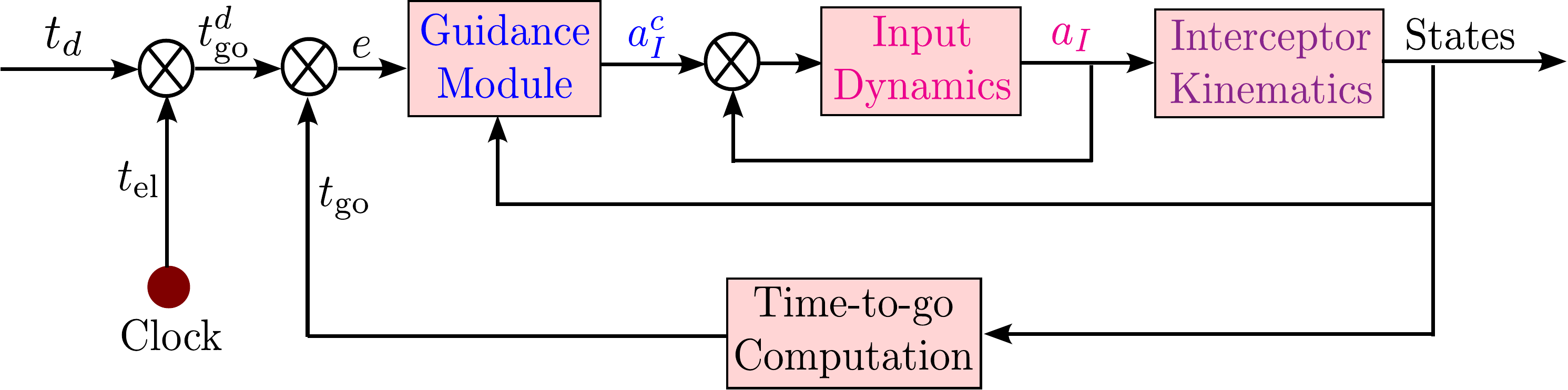}
    \caption{Schematic of the proposed guidance strategy.}
    \label{fig:schematic}
\end{figure}
The detailed schematic representation of the proposed guidance strategy is illustrated in \Cref{fig:schematic}.  Based on the time-to-go error and relevant engagement variables, the guidance module generates the commanded 
acceleration $a_I^c$. To enforce the acceleration constraints, $a_I^c$ is passed through \eqref{eq:saturation_model} to yield a bounded control input satisfying $a_{\min} < a_I < a_{\max}$. This bounded acceleration is subsequently applied to the interceptor kinematics to update the flight path angle. The updated engagement parameters are then fed back to the guidance module and also to compute the time-to-go error, and the loop continues until all desired interception objectives are satisfied.

 \section{Simulations}\label{sec:simulation}
In this section, the efficacy of the proposed guidance strategy is demonstrated through various engagement scenarios. The initial speeds of the interceptor and target are set to $V_I = 70~\mathrm{m/s}$ and $V_T = 50~\mathrm{m/s}$, respectively. The initial separation and LOS angle are taken as $5~\mathrm{km}$ and $0^\circ$, respectively. In all trajectory plots, circle markers denote the launch points while cross markers indicate the interception points. The values of parameters $n$ and $\lambda$ are taken as $1$ and $0.15$ respectively. Moreover, the controller parameters $\alpha$, $\mathcal{M}$, $\theta$, $\kappa$ and $\eta$ are set to $1.2$, $1$, $0.6$, $5$ and $1$ respectively. The bounds on the acceleration are taken as $-4~\mathrm{m/s^2}<a_I<8~\mathrm{m/s^2}$. The value of $c$ is taken as $c=3(V_I+V_T)$.
\subsection{Different Impact Times $t_d$}
\begin{figure}[h!]
    \centering
    \begin{subfigure}[t]{0.475\linewidth}
        \centering
        \includegraphics[width=\linewidth]{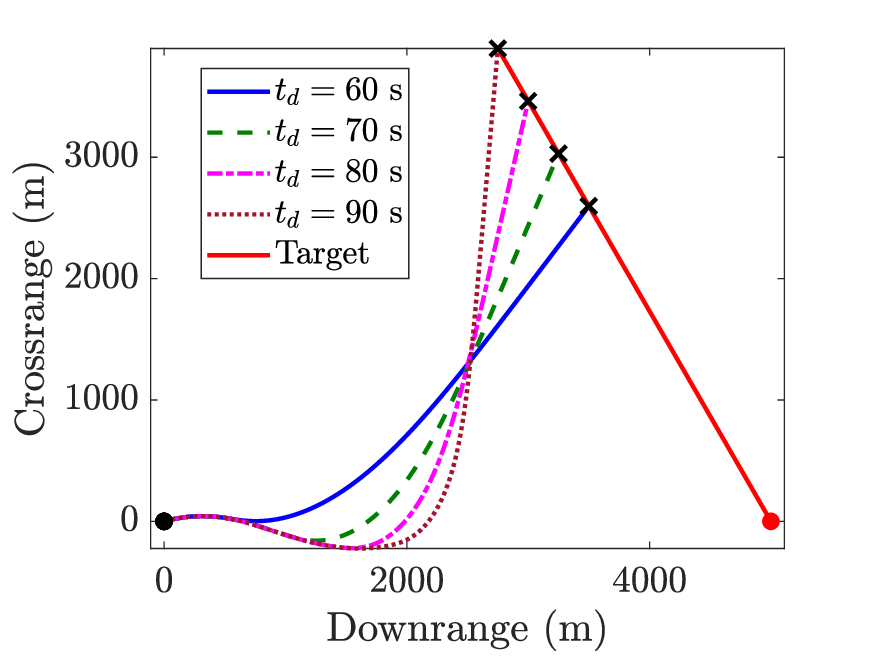}
        \caption{Trajectory.}
        \label{fig:traj1}
    \end{subfigure}
        \begin{subfigure}[t]{0.475\linewidth}
        \centering
        \includegraphics[width=\linewidth]{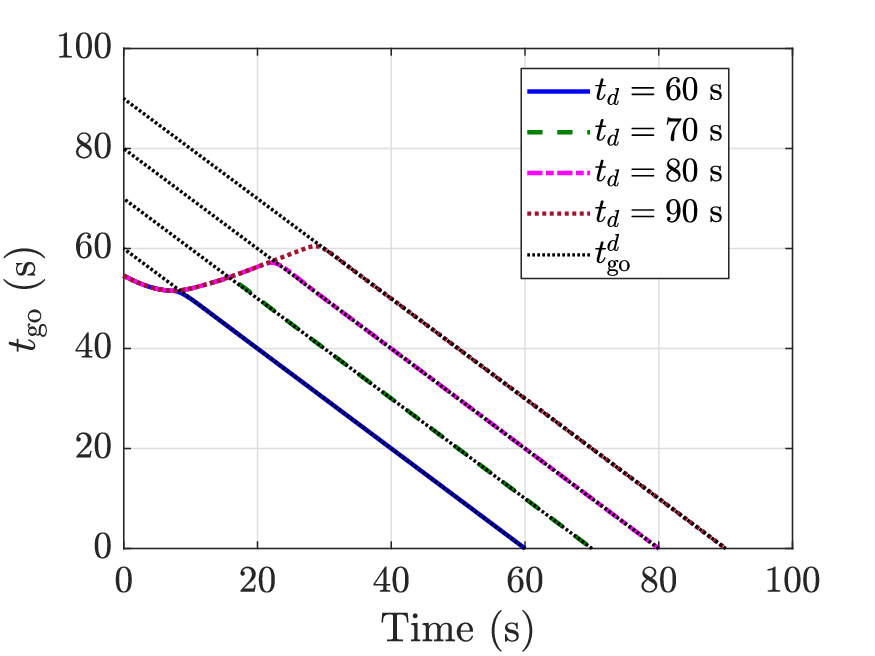}
        \caption{Time-to-go.}
        \label{fig:tgo1}
    \end{subfigure}
        \begin{subfigure}[t]{0.475\linewidth}
        \centering
        \includegraphics[width=\linewidth]{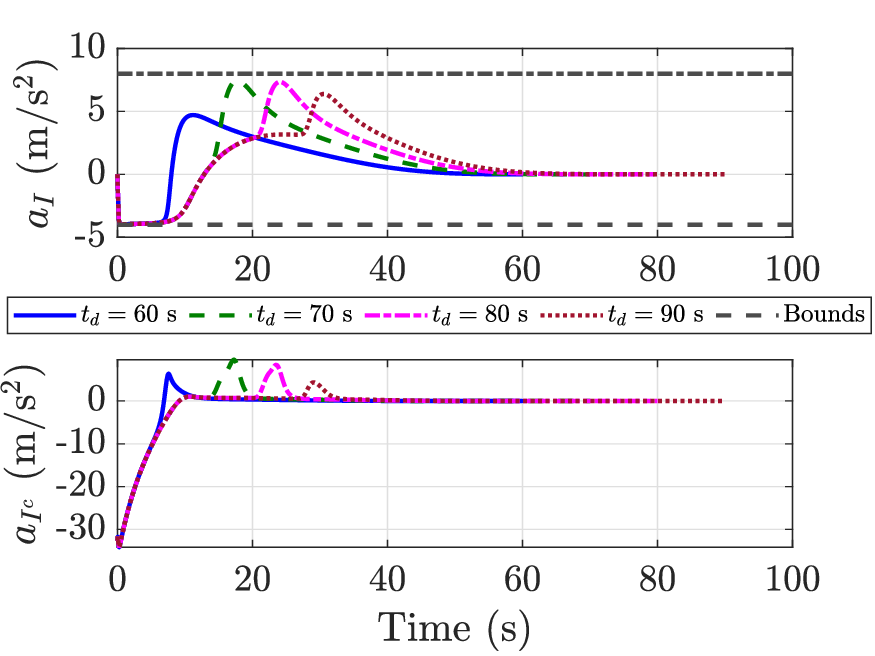}
        \caption{Acceleration ($a_I$ and $a_I^c$).}
        \label{fig:am1}
    \end{subfigure}
    \begin{subfigure}[t]{0.475\linewidth}
        \centering
        \includegraphics[width=\linewidth]{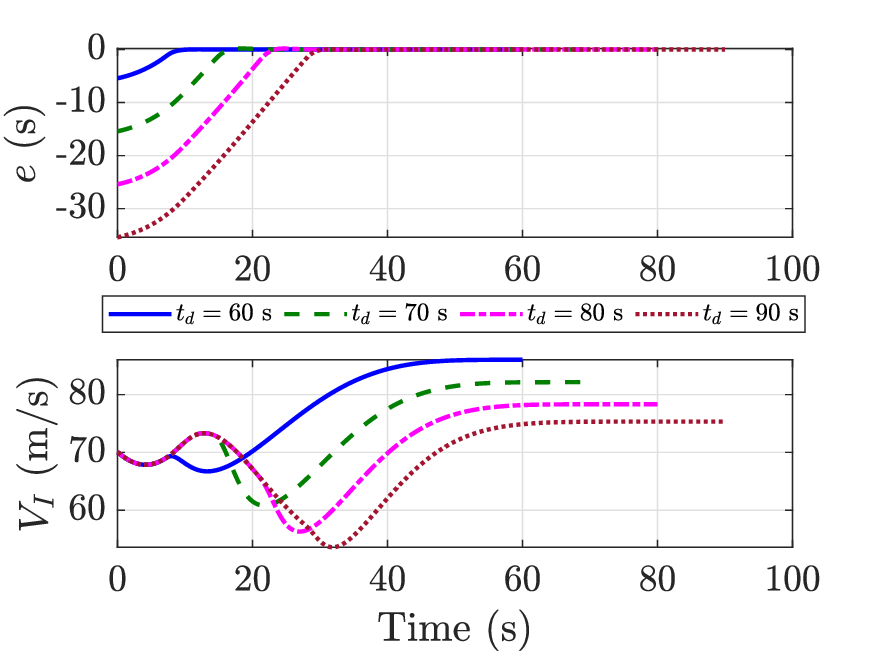}
        \caption{Error ($e$) and velocity.}
        \label{fig:e1}
    \end{subfigure}
\caption{Performance evaluation for constant velocity target interception at different $t_d$.}
    \label{fig:diif_td}
\end{figure}
To begin, a scenario is examined where an interceptor is required to intercept at the target at different impact time  $t_d \in \{60, 70, 80, 90\}~\mathrm{s}$, with initial flight path angles of $\gamma_I = 15^\circ$ and $\gamma_T = 120^\circ$ for the interceptor and target, respectively. The corresponding results are illustrated in \Cref{fig:diif_td}, which demonstrates successful target interception at different specified impact times. The corresponding trajectories are depicted in \Cref{fig:traj1}. The time-to-go profiles are illustrated in \Cref{fig:tgo1}, and it can be observed from \Cref{fig:tgo1} and \Cref{fig:e1} that all errors converge to zero well before $t_d$. From \Cref{fig:am1}, although the commanded acceleration $a_I^c$ exceeds the prescribed bounds, the proposed input saturation dynamics effectively ensure that the achieved acceleration adheres to the prescribed limits $a_{\min} < a_I < a_{\max}$ throughout the engagement. Following error convergence, it can be observed that $a_I$ gradually decreases and eventually converges to zero, after which the velocity remains constant as seen in \Cref{fig:e1}.
\subsection{Different  Initial Heading Angle $\gamma_{I_0}$}
In this subsection, the effectiveness of the proposed strategy is demonstrated for the case of interception of a constant velocity target at $t_d=70~\mathrm{s}$ for different initial interceptor heading angles of $\gamma_{I_0}\in\{5^\circ,10^\circ,20^\circ,30^\circ\}$. All other parameters are taken as the same as in the previous case.
\begin{figure}[h!]
    \centering
    \begin{subfigure}[t]{0.475\linewidth}
        \centering
        \includegraphics[width=\linewidth]{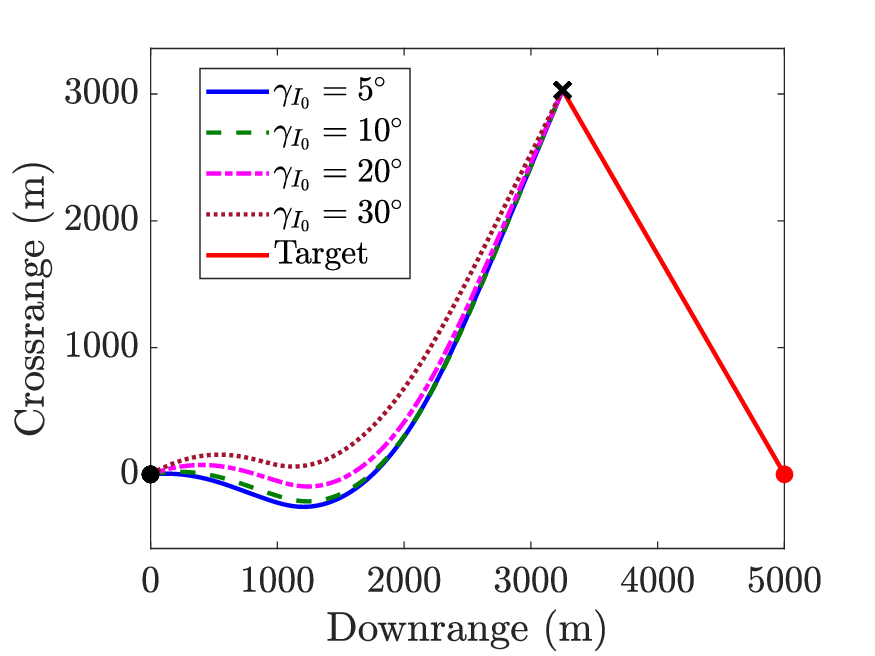}
        \caption{Trajectory.}
        \label{fig:traj2}
    \end{subfigure}
        \begin{subfigure}[t]{0.475\linewidth}
        \centering
        \includegraphics[width=\linewidth]{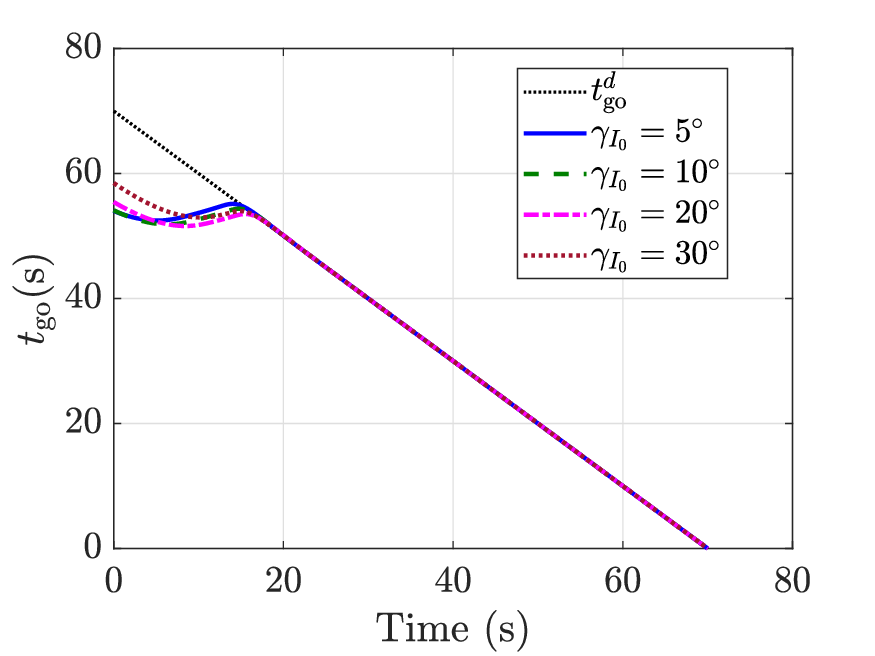}
        \caption{Time-to-go.}
        \label{fig:tgo2}
    \end{subfigure}
        \begin{subfigure}[t]{0.475\linewidth}
        \centering
        \includegraphics[width=\linewidth]{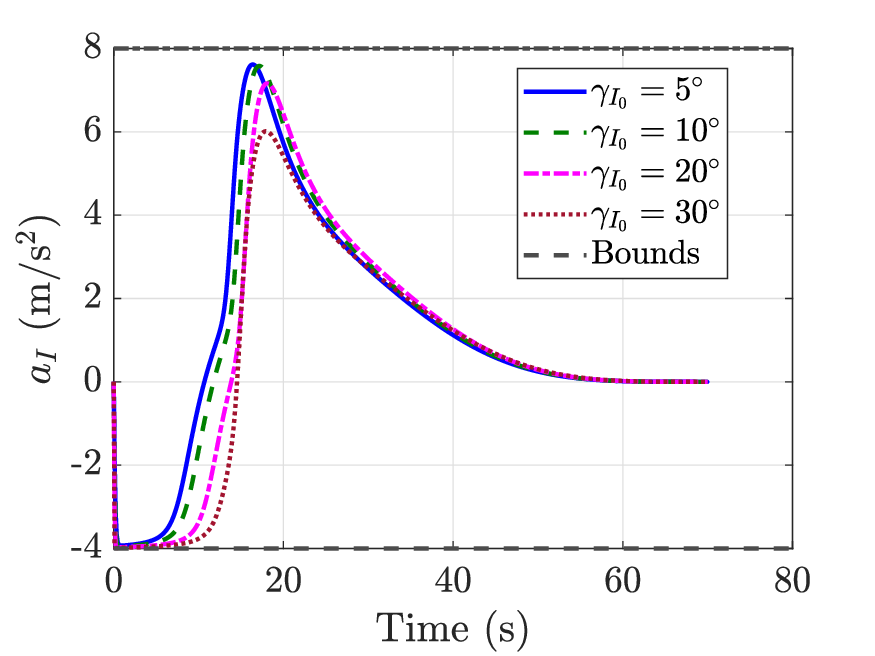}
        \caption{Acceleration.}
        \label{fig:am2}
    \end{subfigure}
    \begin{subfigure}[t]{0.475\linewidth}
        \centering
        \includegraphics[width=\linewidth]{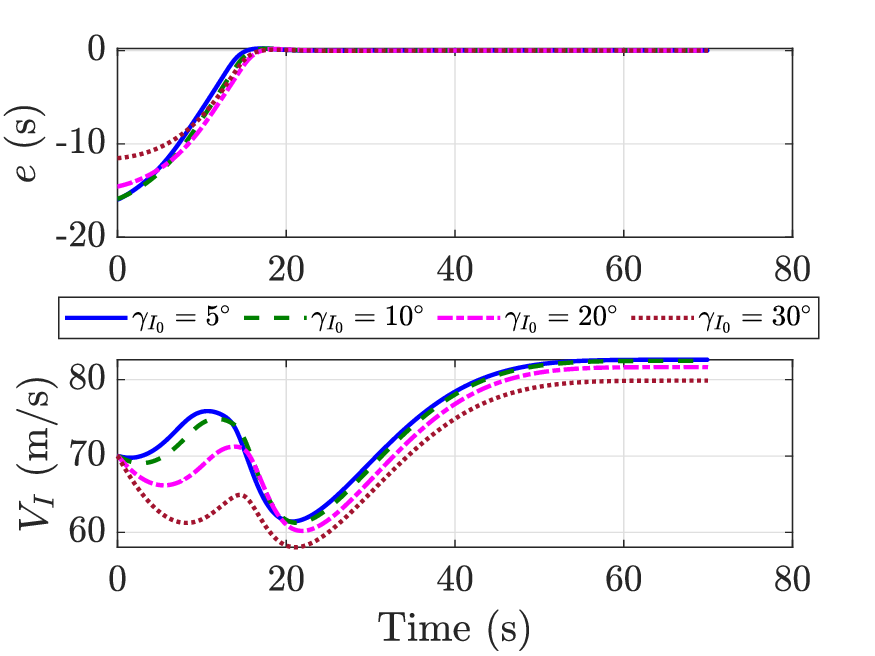}
        \caption{Error ($e$) and velocity.}
        \label{fig:e2}
    \end{subfigure}
\caption{Performance evaluation for constant velocity target interception with different initial heading angles.}
    \label{fig:Diff_Gamma}
\end{figure}

The results are illustrated in \Cref{fig:Diff_Gamma}. From \Cref{fig:traj2}, it can be observed that all interceptors successfully intercept the target at the prescribed time $t_d = 60~\mathrm{s}$ while strictly adhering to the acceleration constraints, as evident from \Cref{fig:am2}. \Cref{fig:tgo2} and \Cref{fig:e2} confirm that the time-to-go error converges to zero well before $t = 22~\mathrm{s} < t_d$ for all interceptors. The corresponding speed profiles depicted in \Cref{fig:e2} exhibit consistent behavior with the previous case.
\subsection{Different Target Heading Angles $\gamma_T$ and Stationary Target }
To further demonstrate the effectiveness of the proposed strategy, simulations are carried out for the interception of a stationary target and a constant-velocity, non-maneuvring target for different heading angles $\gamma_T\in\{70^\circ,90^\circ,135^\circ\}$ at $t_d=130~\mathrm{s}$. The interceptor initial heading angle is taken as $\gamma_{I_0}=15^\circ$. All other parameters remain the same as in the previous cases.
\begin{figure}[h!]
    \centering
    \begin{subfigure}[t]{0.475\linewidth}
        \centering
        \includegraphics[width=\linewidth]{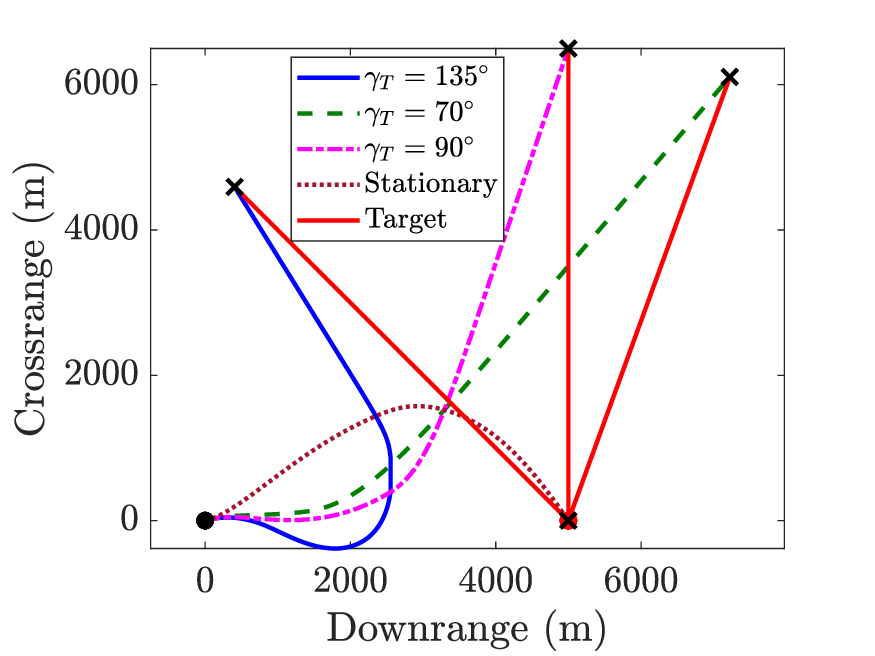}
        \caption{Trajectory.}
        \label{fig:traj3}
    \end{subfigure}
        \begin{subfigure}[t]{0.475\linewidth}
        \centering
        \includegraphics[width=\linewidth]{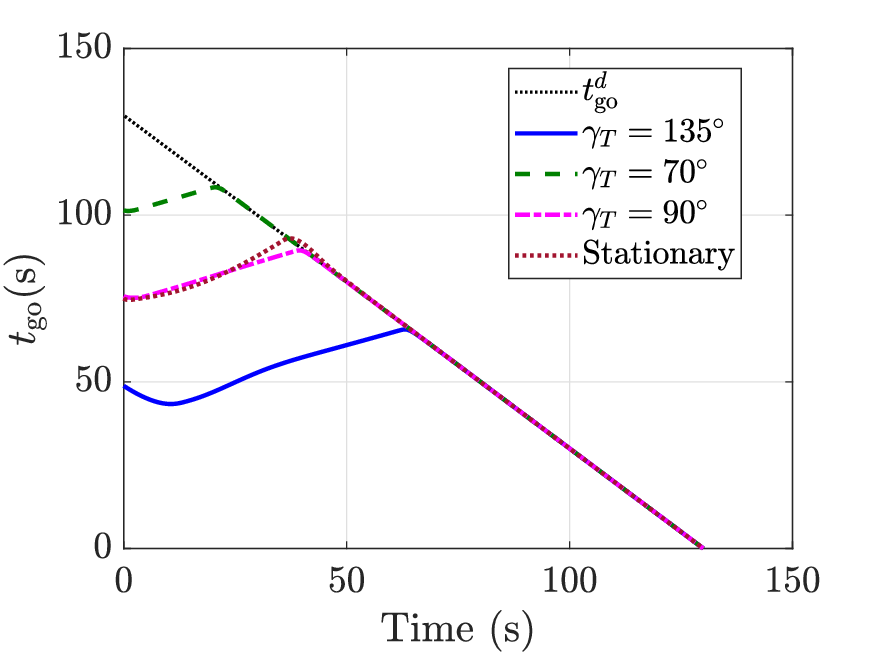}
        \caption{Time-to-go.}
        \label{fig:tgo3}
    \end{subfigure}
        \begin{subfigure}[t]{0.475\linewidth}
        \centering
        \includegraphics[width=\linewidth]{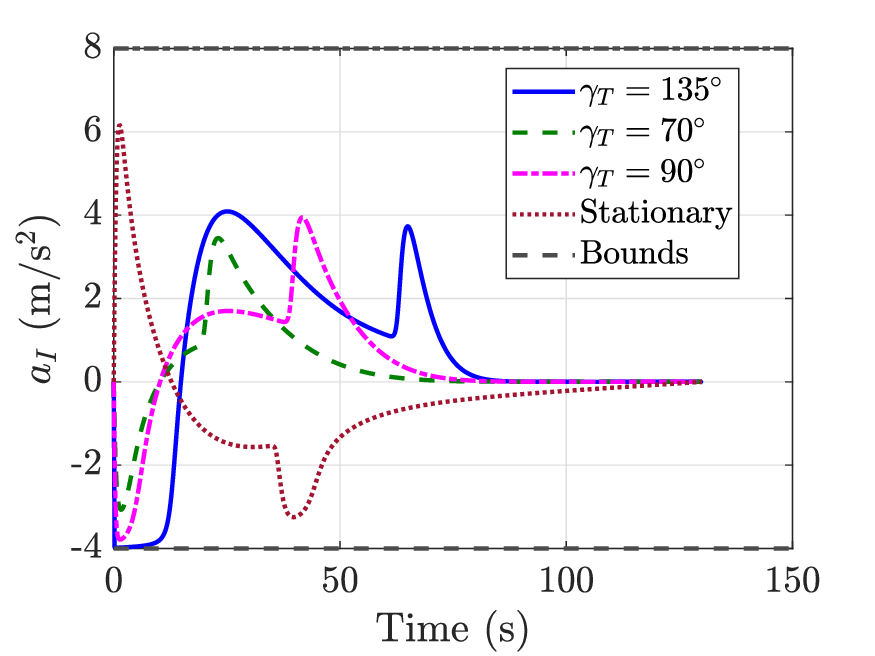}
        \caption{Acceleration.}
        \label{fig:am3}
    \end{subfigure}
    \begin{subfigure}[t]{0.475\linewidth}
        \centering
        \includegraphics[width=\linewidth]{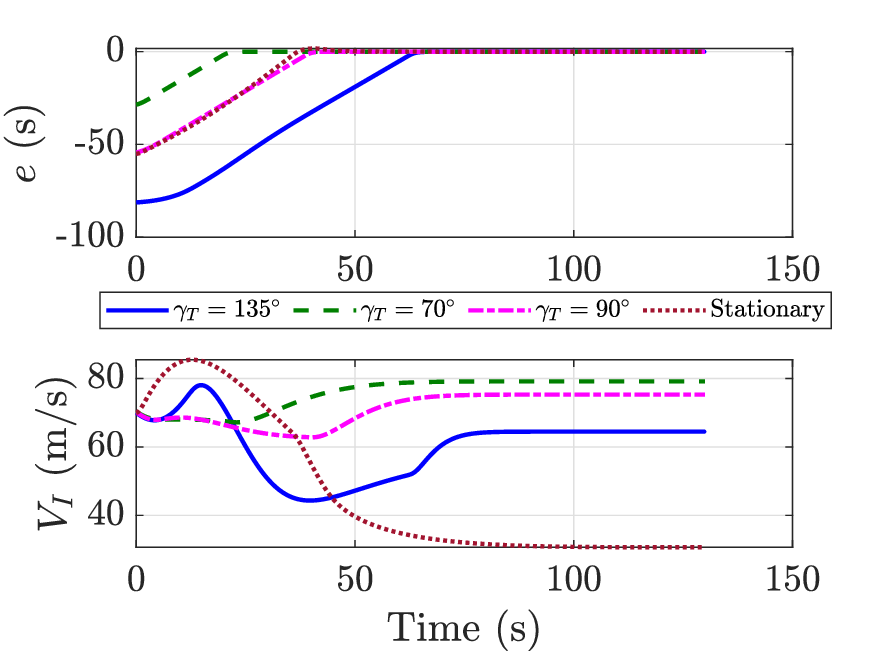}
        \caption{Error ($e$) and velocity.}
        \label{fig:e3}
    \end{subfigure}
\caption{Performance evaluation for interception of a stationary target and a constant velocity target with different heading angles.}
\label{fig:diff_GT}
\end{figure}
\Cref{fig:diff_GT} illustrates the obtained results. From \Cref{fig:traj3}, successful interception at the prescribed time $t_d = 130~\mathrm{s}$ is observed across all cases. \Cref{fig:tgo3} and \Cref{fig:e3} confirm that the time-to-go errors converges to zero well ahead of $t_d$.  The corresponding acceleration profiles are presented in \Cref{fig:am3}, from which it is evident that the acceleration constraints are strictly adhered to throughout the engagement. The speed profiles in \Cref{fig:e3} follow trends consistent with the previous scenario, with the stationary target case exhibiting a gradual reduction in speed that stabilizes to a constant value as the acceleration demand approaches zero.

\subsection{Comparison}
To further validate the efficacy of the proposed guidance strategy, a comparative evaluation is performed against the method presented in \cite{kumar2022true} for the interception of a constant-velocity, non-maneuvering target at $t_d = 55~\mathrm{s}$. The interceptor and target speeds are set to $V_I = 250~\mathrm{m/s}$ and $V_T = 200~\mathrm{m/s}$, respectively, consistent with the high-speed engagement scenario considered in \cite{kumar2022true}, with an initial separation of $r = 10~\mathrm{km}$. To ensure a meaningful quantitative comparison, the gains of both strategies are tuned to yield approximately the same convergence time. The gains $\mathcal{M}_1$, $\mathcal{M}_2$, and $\alpha$ for the method in \cite{kumar2022true} are set to $0.2$, $0.2$, and $51/53$, respectively, while the gains $\theta$ and $\kappa$ for the proposed strategy are set to $0.5$ and $10$, respectively. The acceleration bounds are taken as $-6~\mathrm{g} < a_I < 6~\mathrm{g}$, where $\mathrm{g}$ denotes the acceleration due to gravity. All remaining parameters are kept the same as in the previous scenario. The standard-reaching law is also included in the comparison to demonstrate the effectiveness of the considered exponential-reaching law.

\begin{figure}[h!]
    \centering
    \begin{subfigure}[t]{0.475\linewidth}
        \centering
        \includegraphics[width=\linewidth]{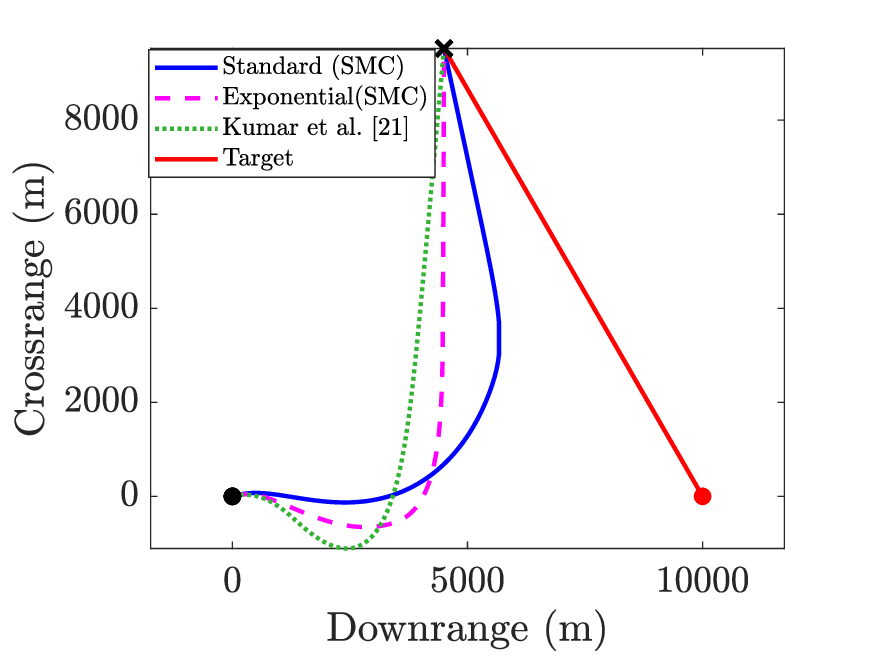}
        \caption{Trajectory.}
        \label{fig:traj5}
    \end{subfigure}
        \begin{subfigure}[t]{0.475\linewidth}
        \centering
        \includegraphics[width=\linewidth]{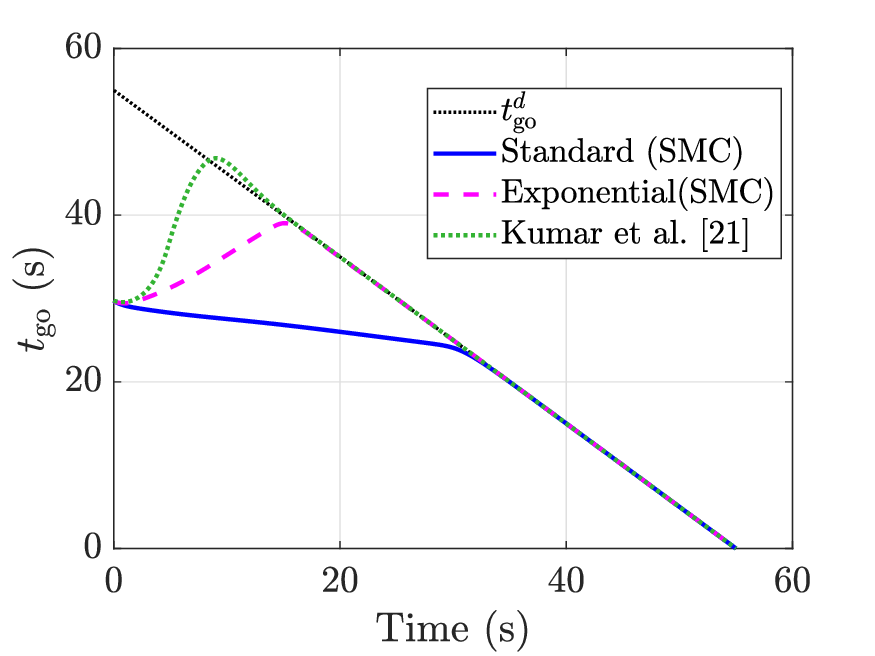}
        \caption{Time-to-go.}
        \label{fig:tgo5}
    \end{subfigure}
        \begin{subfigure}[t]{0.475\linewidth}
        \centering
        \includegraphics[width=\linewidth]{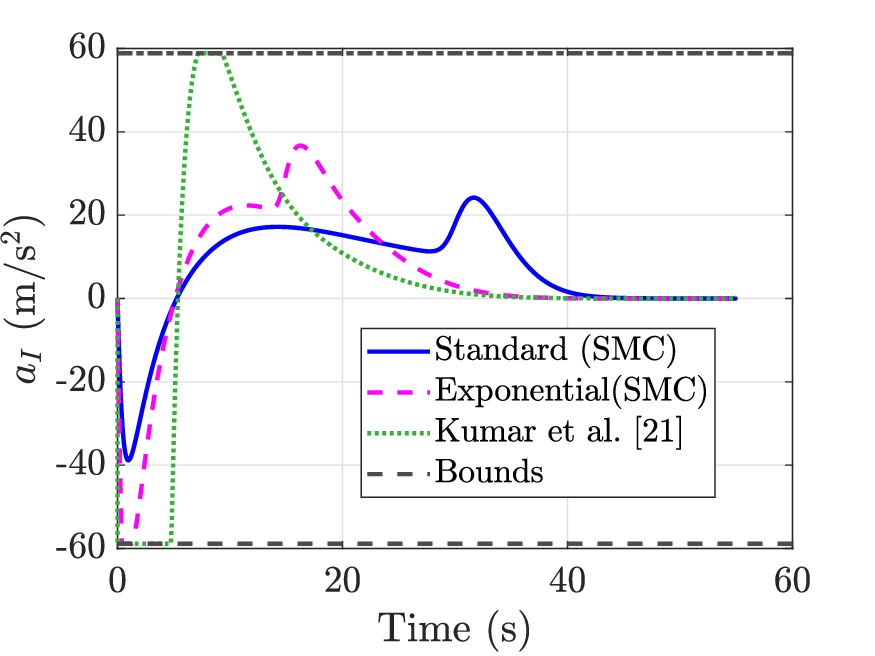}
        \caption{Acceleration.}
        \label{fig:am5}
    \end{subfigure}
    \begin{subfigure}[t]{0.475\linewidth}
        \centering
        \includegraphics[width=\linewidth]{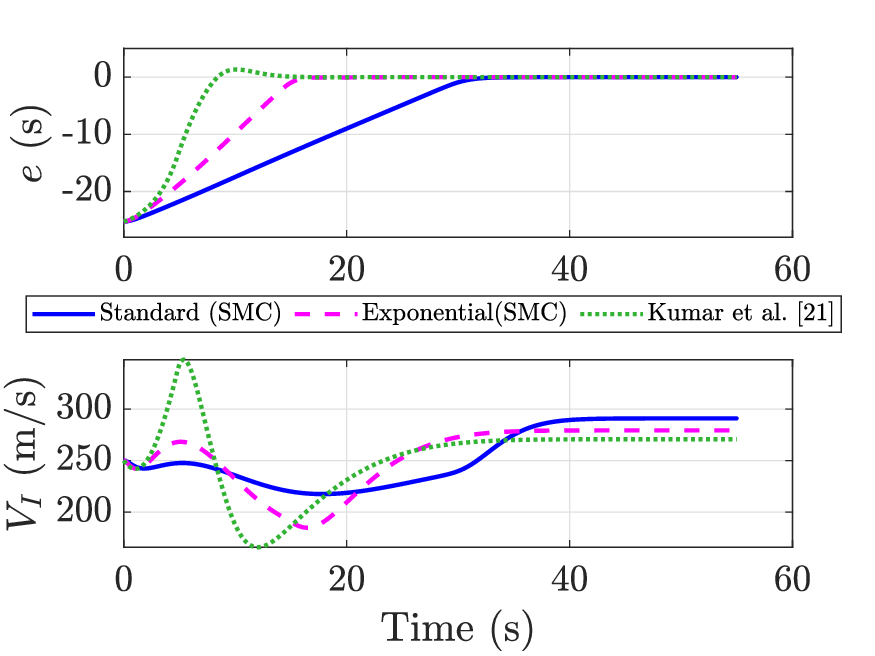}
        \caption{Error ($e$) and velocity.}
        \label{fig:e5}
    \end{subfigure}
\caption{Comparison of the proposed guidance strategy with the strategy in \cite{kumar2022true}.}
\label{fig:comparison}
\end{figure}

The comparative results are illustrated in \Cref{fig:comparison}. As evident from \Cref{fig:traj5}, the interceptor successfully intercepts the target at the prescribed time $t_d = 55~\mathrm{s}$ under all considered guidance strategies. It is clear from \Cref{fig:tgo5} and \Cref{fig:e5} that the considered exponential reaching converges faster than the standard reaching law for the same value of $\mathcal{M}$, confirming the result \eqref{eq:tr_diff}. Moreover, one can also observe that the convergence times of the strategy in \cite{kumar2022true} and the proposed strategy are approximately the same.  The corresponding acceleration and speed profiles are illustrated in \Cref{fig:am5} and \Cref{fig:e5}, respectively. It can be observed that the proposed guidance strategies based on the standard and exponential reaching laws exhibit smoother acceleration profiles compared to \cite{kumar2022true}, owing to the explicit incorporation of bounded control inputs within the guidance design, unlike the unbounded input assumption adopted in \cite{kumar2022true}. For a quantitative comparison, the integral control effort is computed and tabulated in \Cref{tab:control_comparison}.
\begin{table}[h!]
\centering
\caption{Comparison of convergence time and control effort}
\label{tab:control_comparison}
\begin{tabular}{lcc}
\hline
\textbf{Case} & \textbf{Convergence Time (s)} & $\displaystyle J=\int_{0}^{t_d} a_I^2 \, dt$ (m$^2$/s$^3$) \\
\hline
Standard SMC        & 34.0 & 10152.5848 \\
Exponential SMC     & 17.1 & 18238.4510 \\
Kumar et al.~\cite{kumar2022true} & 15.5 & 38985.8903 \\
\hline
\end{tabular}
\end{table}

From \Cref{tab:control_comparison}, it can be observed that the exponential reaching law exhibits a $79\%$ higher control effort compared to the standard reaching law. This can be attributed to the faster convergence of the exponential reaching law, which achieves convergence in approximately half the time required by the standard reaching law. Compared to \cite{kumar2022true}, the proposed strategy demonstrates a notable $53\%$ reduction in control effort. This significant reduction highlights the benefit of explicitly incorporating input bounds within the guidance design, which prevents excessive control commands, results in more energy-efficient engagement, and guarantees overall closed-loop behavior.
\section{Conclusions and Future Work}\label{sec:conclusions}
A nonlinear guidance strategy for intercepting a constant-velocity, non-maneuvering target with input bounds explicitly incorporated within the guidance design was developed in the present work. By adopting TPNG with an exact time-to-go formulation, the proposed approach avoids linearization and small-angle assumptions, ensuring applicability for a wide range of target motion. A first-order, input-affine, asymmetric, and differentiable saturation model was employed to embed input bounds into a sliding mode-based guidance law, enabling accurate time-constrained interception with bounded control and guaranteed closed-loop stability. Simulation results validated the effectiveness of the proposed strategy across multiple engagement scenarios. Comparison with existing methods demonstrates that the proposed guidance design yields more energy-efficient engagements while guaranteeing closed-loop stability. Future work will focus on extending the proposed framework to maneuvering targets and incorporating rate constraints into the guidance design.

    \bibliographystyle{IEEEtran}
    \bibliography{references_aacc}
%\bibliography{Hybrid_Guidance/references_h,Stationary_2D_DPTO/references}

\end{document}